\documentclass[11pt]{article}
\usepackage[T1]{fontenc}
\usepackage{lmodern}
\usepackage{amsmath,amssymb,amsthm}
\usepackage[margin=1.1in]{geometry}
\usepackage{fancyhdr}
\usepackage{hyperref}
\hypersetup{colorlinks=true,linkcolor=black,citecolor=blue,urlcolor=blue,
  pdftitle={Bounds on the real tensor rank of octonion multiplication},
  pdfauthor={Hardik Jain}}

\fancypagestyle{firstpage}{%
  \fancyhf{}%
  \fancyfoot[R]{\small\thepage}%
}

\theoremstyle{plain}
\newtheorem{theorem}{Theorem}
\newtheorem{lemma}[theorem]{Lemma}
\newtheorem{proposition}[theorem]{Proposition}
\theoremstyle{definition}
\newtheorem{remark}[theorem]{Remark}

\newcommand{\R}{\mathbb{R}}
\newcommand{\C}{\mathbb{C}}
\newcommand{\HH}{\mathbb{H}}
\newcommand{\OO}{\mathbb{O}}
\newcommand{\rank}{\mathrm{R}}
\newcommand{\Nm}{\mathrm{N}}
\DeclareMathOperator{\rk}{rank}

\begin{document}

\title{\bfseries Bounds on the real tensor rank of octonion multiplication}

\author{Hardik Jain\\[3pt]
  \normalsize\href{mailto:mxthxmxtxcs@outlook.com}{\texttt{mxthxmxtxcs@outlook.com}}}
\date{}

\maketitle
\thispagestyle{firstpage}

\begin{abstract}
The tensor rank of a bilinear map is the least number of multiplications any
bilinear algorithm needs to compute it; for the multiplication of an algebra
it measures how cheaply the algebra can be multiplied at all. For the
even-dimensional real normed division algebras it is $3$ for the complex numbers
and $8$ for the quaternions, both classical, while for the octonions $\OO$
only a range was known: at least $15$ (Fiduccia and Zalcstein, 1977) and at
most $30$ (Cariow and Cariowa). We prove
\[
  18 \;\le\; \rank_\R(T_\OO) \;\le\; 25 .
\]
The lower bound peels the eight slices of $T_\OO$ down to two and bounds the
rank of the surviving pencil through the octonion norm. Nothing in it is
special to dimension $8$: the same steps give $\rank_\R(T_A)\ge\frac52 n-2$
for every real normed division algebra $A$ of even dimension $n$, sharp for $\C$
and $\HH$ and the best bound we know for $\OO$. The upper bound is a separate
construction, an explicit rank-$25$ decomposition certified by a Krawczyk
argument, in exact rational arithmetic, to sit within $10^{-6}$ of an exact
one. The same two arguments pin down the rank of a smaller three-slice
quaternion tensor $\tau$, giving $\rank_\R(\tau)=7$. The Lean~4 kernel checks
the lower bounds and the Krawczyk existence principle; the accompanying
scripts check the certificate's finitely many exact-rational inequalities.
\end{abstract}

\begin{center}
The formalization, the certificates, and their verification are available at\\
\url{https://github.com/hxrdxkxvxd/octonion-rank}.
\end{center}

\section{Introduction}

A bilinear algorithm computes a bilinear map $\beta$ by forming a number of
products, each of a linear form in the first argument with a linear form in the
second, and returning linear combinations of those products. Additions and
multiplications by fixed scalars are not counted; the products are, and the
least number of them that suffices is the tensor rank of $\beta$, which
coincides with the rank of the order-$3$ tensor associated with $\beta$
\cite{Strassen73,BCS,Landsberg}. Asking for this number for the multiplication
map of an algebra is asking how cheaply that algebra can be multiplied at all,
by any method of this kind, and the answer is a property of the algebra rather
than of any particular scheme for computing in it. For the real normed
division algebras the answer is classical in low dimension: complex multiplication has
rank $3$ \cite{deGroote}, quaternion multiplication rank $8$
\cite{HowellLafon}. For
the octonions $\OO$ \cite{Baez,SpringerVeldkamp} it is not known.

What is known is a range. The naive algorithm uses $64$ multiplications;
expanding the Cayley--Dickson formula
$(a,b)(c,d)=(ac-\bar d b,\, da+b\bar c)$ and computing each of its four
quaternion products by an optimal eight-multiplication algorithm uses $32$,
as does the dedicated octonion algorithm of Cariow and Cariowa~\cite{Cariow2012};
their unified method for hypercomplex multiplication~\cite{Cariow} uses $30$. In the
other direction, Fiduccia and Zalcstein \cite{FiducciaZalcstein} proved in 1977
that every finite-dimensional real algebra without zero divisors has
multiplicative complexity at least $2n-1$ (their Theorem~6), recording the case
of the Cayley numbers explicitly as $\ge 15$ (their Example~4); multiplicative
complexity is counted over general arithmetic chains and so lower-bounds the
rank. We narrow the range on both sides, to
\[
  18\ \le\ \rank_\R(T_\OO)\ \le\ 25 .
\]
Neither bound is known to be sharp, and closing the gap remains open.

The lower bound comes from a substitution, or peeling, argument that removes
the tensor's slices one at a time, each worth a unit of rank, stopped one step
before its natural end and completed by a rank estimate for the two-slice
pencil that remains.
The argument uses nothing specific to dimension $8$: the same steps bound
$\rank_\R(T_A)$ for every even-dimensional real normed division algebra $A$
(Theorem~\ref{thm:uniform}), and the octonion gain over Fiduccia--Zalcstein is
one instance of a uniform improvement. The upper bound is obtained
differently, by exhibiting a numerical rank-$25$ decomposition and certifying
that an exact one lies near it.

The lower bounds, the pencil estimate, and the certification principle behind
the upper bound are formalized in the Lean~4 proof assistant
\cite{Lean4,Mathlib} and accepted by its kernel. The certificate itself, some
$2.6\times10^{5}$ rational numbers, is checked by the scripts in the
repository rather than inside Lean; the boundary between the two is stated
where it arises.

\subsection*{Notation and conventions}
Let $\OO$ be the real octonion algebra with standard basis $e_0,\dots,e_7$,
$e_0=1$, norm form $\Nm(x)=\sum_i x_i^2$, and conjugation $x\mapsto\bar x$.
For $x\in\OO$ let $L_x\colon\OO\to\OO$ denote left multiplication, an
$\R$-linear endomorphism of $\OO\cong\R^8$; the composition property of the
norm gives $L_{\bar x}L_x=\Nm(x)\,I$. We work with the structure tensor
$T_\OO$ through its first-factor slices, the matrices of the operators
$L_{e_0},\dots,L_{e_7}$; tensor rank is unchanged by permuting or transposing
the three factors, so no generality is lost in this choice. For a finite family
$S=(S_a)$ of $n\times n$ real matrices, the real rank $\rank_\R(S)$ is the
least $r$ for which there exist $u_1,\dots,u_r,v_1,\dots,v_r\in\R^n$ and
scalars $f_{s,a}$ with $S_a=\sum_{s=1}^{r} f_{s,a}\,u_s v_s^{\top}$ for every
$a$; this is the tensor rank of the associated order-$3$ tensor \cite{BCS}. It
is unchanged under an invertible linear substitution in any one of the three
factors: replacing the family by $\big(\sum_b g_{ab}S_b\big)$ for an invertible
matrix $(g_{ab})$, or every $S_a$ by $PS_aQ$ for invertible $P,Q$, leaves
$\rank_\R$ the same. We use both normalizations on the surviving pencil below.

\section{The lower bound}

\begin{theorem}\label{thm:main-lower}
$\rank_\R(T_\OO)\ge 18$.
\end{theorem}

The proof combines a substitution (peeling) principle with a rank estimate
for two-term pencils. We isolate the pencil estimate first, as it carries the
essential content and is stated for general even dimension so that it applies
without change to the quaternion and complex cases.

\begin{theorem}\label{thm:pencil}
Let $n$ be even and let $C\in\R^{n\times n}$ satisfy $C^2-2aC+bI=0$ for some
$a,b\in\R$ with $a^2<b$. Then the two-term pencil $(I,C)$ has real rank at
least $n+n/2$.
\end{theorem}

Informally: the hypothesis $a^2<b$ makes $C$, after an affine change, a
complex structure $J$ with $J^2=-I$. A rank-$r$ decomposition
of $(I,J)$ produces an $r\times n$ matrix $D$ whose rank measures how far the
decomposition is from respecting $J$. One constraint caps that rank at $r-n$,
the amount by which $r$ exceeds $n$; the other is that $J$ carries the kernel of
$D$ off itself, confining the kernel to half of $\R^n$ and forcing
$\rk D\ge n/2$. The two meet at $r\ge n+n/2$.

\begin{proof}
Set $c=\sqrt{b-a^2}>0$ and $J=c^{-1}(C-aI)$; the hypotheses give $J^2=-I$.
Suppose $(I,C)$ has a rank-$r$ decomposition. Passing to $J$ by an invertible
scalar change, one obtains simultaneous representations
\[
  I=\sum_{s=1}^{r}\alpha_s\,u_s v_s^\top,\qquad
  J=\sum_{s=1}^{r}\beta_s\,u_s v_s^\top ,
\]
with common vectors $u_s,v_s$. Form the matrix $U=[u_1\,\cdots\,u_r]$, let
$A,B$ be the $r\times n$ matrices with rows $\alpha_s v_s^\top,\beta_s
v_s^\top$, and put $D=A+BJ$. The two representations give $UA=I$ and $UB=J$,
hence
\[
  UD=UA+(UB)J=I+J^2=0 .
\]
As $U$ has the right inverse $A$, its kernel has dimension $r-n$ and contains
the range of $D$, so $\rk D\le r-n$. On the other hand, the $s$th row of $D$
is $\alpha_s v_s^\top+\beta_s v_s^\top J$, so for $w$ with $Dw=0$ and
$D(Jw)=0$,
\[
  \alpha_s\langle v_s,w\rangle+\beta_s\langle v_s,Jw\rangle=0,\qquad
  \alpha_s\langle v_s,Jw\rangle-\beta_s\langle v_s,w\rangle=0,
\]
the second because $J^2=-I$. Multiplying the first by $\alpha_s$, the second
by $\beta_s$, and subtracting gives
$(\alpha_s^2+\beta_s^2)\langle v_s,w\rangle=0$, whence
$\alpha_s\langle v_s,w\rangle=0$ for all $s$ and therefore
$w=Iw=\sum_s\alpha_s\langle v_s,w\rangle u_s=0$. Thus
$\ker D\cap J(\ker D)=0$ (note $J^{-1}=-J$, so $J(\ker D)$ is exactly the set
of $w$ with $Jw\in\ker D$); as $J$ is invertible the two subspaces have equal
dimension, so $\dim\ker D\le n/2$, $\rk D\ge n/2$, and
\[
  r\ \ge\ n+\rk D\ \ge\ n+n/2. \qedhere
\]
\end{proof}

The argument is coordinate-free and uses neither the Kronecker classification
of pencils nor any complex module structure. A bound of this type is also a
consequence of the invariant-factor analysis of Sumi, Miyazaki, and
Sakata~\cite{SMS}. The second ingredient is the substitution (peeling) lemma;
the standard pivot proof works over any field.

\begin{lemma}[{Substitution; cf.\ \cite[Prop.~3.1]{LM}}]\label{lem:subst}
Let $S_1,\dots,S_{m+1}$ be $n\times n$ real matrices with $S_{m+1}\neq 0$.
Then there exist $c_1,\dots,c_m\in\R$ such that
\[
  \rank_\R\big(S_1+c_1S_{m+1},\,\dots,\,S_m+c_mS_{m+1}\big)
  \ \le\ \rank_\R\big(S_1,\dots,S_{m+1}\big)-1 .
\]
\end{lemma}

\begin{proof}[Proof of Theorem~\ref{thm:main-lower}]
Every nonzero linear combination of the slices $L_{e_p}$ is again $L_x$ for a
nonzero $x\in\OO$, so the slices form a family of eight left multiplications by
linearly independent octonions $x_1,\dots,x_8$ (the basis $e_0,\dots,e_7$ in
some order). Apply Lemma~\ref{lem:subst} to the last slice. Whatever the
constants $c_i$, the modified slices $L_{x_i}+c_iL_{x_8}=L_{x_i+c_ix_8}$
$(i=1,\dots,7)$ form a family of seven left multiplications by octonions that
are again linearly independent, so the lemma applies once more; after six
applications,
\[
  \rank_\R(T_\OO)\ \ge\ 6 + \rank_\R\big(L_u,L_v\big)
\]
for some linearly independent $u,v\in\OO$.

It remains to bound the residual pencil. Put $M=L_{\bar u}L_v$. Polarizing
$L_{\bar x}L_x=\Nm(x)I$ gives $L_{\bar x}L_y+L_{\bar y}L_x=2\langle
x,y\rangle I$; applied to the pair $(\bar v,\bar u)$, and using that
conjugation preserves the inner product, this reads
$L_vL_{\bar u}=2\langle u,v\rangle I-L_uL_{\bar v}$. Hence, by associativity
of composition of linear maps (no associativity of $\OO$ enters),
\[
  M^2=L_{\bar u}\big(L_vL_{\bar u}\big)L_v
     =2\langle u,v\rangle\,M-\big(L_{\bar u}L_u\big)\big(L_{\bar v}L_v\big)
     =2\langle u,v\rangle\,M-\Nm(u)\Nm(v)\,I .
\]
The discriminant of the quadratic $t^2-2\langle u,v\rangle t+\Nm(u)\Nm(v)$ is
$4\big(\langle u,v\rangle^2-\Nm(u)\Nm(v)\big)$, strictly negative by the
strict Cauchy--Schwarz inequality since $u,v$ are independent. As $L_u$ is
invertible (from $L_{\bar u}L_u=\Nm(u)I$ with $\Nm(u)>0$), the matrix
$C:=L_u^{-1}L_v=\Nm(u)^{-1}M$ satisfies an irreducible real quadratic, and
the pencils $(L_u,L_v)$ and $(I,C)$ have the same rank.
Theorem~\ref{thm:pencil} with $n=8$ gives $\rank_\R(L_u,L_v)\ge 12$, whence
\[
  \rank_\R(T_\OO)\ \ge\ 6+12\ =\ 18 . \qedhere
\]
\end{proof}

The proof used no property of $\OO$ beyond three facts: the composition
identity $L_{\bar x}L_x=\Nm(x)I$, the positive-definiteness of $\Nm$, and
$\dim\OO=8$. Replacing $8$ by $n$ throughout gives a uniform statement.

\begin{theorem}\label{thm:uniform}
Let $A$ be a real unital composition algebra (a Hurwitz algebra) of even
dimension $n$ whose norm form is positive definite. Then
\[
  \rank_\R(T_A)\ \ge\ \tfrac{5}{2}n-2 .
\]
\end{theorem}

\begin{proof}
The slices of $T_A$ are the $n$ operators $L_{e_0},\dots,L_{e_{n-1}}$, and
every nonzero linear combination of them is $L_x$ for some $x\neq 0$, hence
invertible. Lemma~\ref{lem:subst} applies $n-2$ times, each application
preserving both the form of the family and the linear independence of the
underlying algebra elements, and leaves a pencil $(L_u,L_v)$ with $u,v$
independent. As in the proof of Theorem~\ref{thm:main-lower}, $C=L_u^{-1}L_v$
satisfies a real quadratic with negative discriminant, so
Theorem~\ref{thm:pencil} gives $\rank_\R(L_u,L_v)\ge n+n/2$. Hence
$\rank_\R(T_A)\ge(n-2)+n+n/2=\frac52 n-2$.
\end{proof}

By Hurwitz's theorem the algebras satisfying the hypotheses are $\C$, $\HH$ and
$\OO$, of dimensions $2$, $4$ and $8$, so Theorem~\ref{thm:uniform} reads
\[
  \rank_\R(T_\C)\ge 3,\qquad \rank_\R(T_\HH)\ge 8,\qquad
  \rank_\R(T_\OO)\ge 18 .
\]
The first two are the exact ranks \cite{deGroote,HowellLafon}, so on the
classical cases the bound is sharp. It also exceeds the Fiduccia--Zalcstein
bound $2n-1$ for every $n>2$: for $\HH$ it gives the sharp value $8$ where
$2n-1=7$ falls short, and for $\OO$ it gives $18$ against $15$.

It is worth isolating where that gain comes from. Peeling all the way down to a
single invertible slice, of rank $n$, would give only $(n-1)+n=2n-1$, exactly
the Fiduccia--Zalcstein bound \cite{FiducciaZalcstein}, which holds for any
finite-dimensional algebra without zero divisors and uses no associativity. In
the associative world this coincides with the Alder--Strassen bound
\cite{AlderStrassen} for an algebra with a single maximal two-sided ideal.
Halting one step earlier, at the two-slice pencil, and invoking
Theorem~\ref{thm:pencil} instead, replaces the final $n$ by $n+n/2$ and so adds
$n/2-1$; for $\OO$ that is the difference between $15$ and $18$.

The same two ingredients also settle a tensor that is not the multiplication
tensor of a composition algebra. Let $\tau$ be the tensor of quaternion
multiplication with the left factor restricted to the span of $1$,
$i$ and $j$, so that its slices are the left multiplications $(I,L_i,L_j)$ on
$\HH\cong\R^4$: one substitution, then Theorem~\ref{thm:pencil} with $n=4$,
which contributes $4+4/2=6$, gives $\rank_\R(\tau)\ge 7$. Unlike the octonion
case it is matched by a certificate below, so the rank of $\tau$ is determined
exactly.

Theorems~\ref{thm:main-lower} and~\ref{thm:pencil} and the bound
$\rank_\R(\tau)\ge 7$ are formalized in Lean~4 \cite{Lean4} over the mathlib
library \cite{Mathlib}. The peeling step underlying Theorem~\ref{thm:uniform}
is formalized in its parametric form, for an arbitrary injective linear family
$\Phi$ whose pencils have rank at least $n+n/2$; what is discharged inside the
kernel for a particular algebra is that pencil hypothesis, which is done for
$\OO$ and for $\tau$. The cases $\C$ and $\HH$ of
Theorem~\ref{thm:uniform} are recorded above as consequences of the same
argument and are not separately formalized. The octonion multiplication table
is given explicitly, and the two facts the proof uses (the composition identity
$L_{\bar x}L_x=\Nm(x)I$ and its polarization) are proved from it as theorems,
quantified over all $x$; together with the unit property, these characterize
$\OO$ up to isomorphism by Hurwitz's theorem (Remark~\ref{rem:composition};
cf.\ \cite{SpringerVeldkamp}). The
final proof terms are accepted by Lean's kernel and depend only on its three
foundational axioms: propositional extensionality, choice, and quotient
soundness.

\begin{remark}\label{rem:composition}
The proof of Theorem~\ref{thm:main-lower} uses the octonions only through two
features: the composition identity $L_{\bar x}L_x=\Nm(x)I$, together with the
invertibility of $L_x$ for $x\neq 0$, and the positive-definiteness of $\Nm$,
which enters through strict Cauchy--Schwarz. By Hurwitz's theorem
\cite{SpringerVeldkamp} the octonions are the unique eight-dimensional real
unital composition algebra with positive definite norm, so within that class
the statement is specific to $\OO$.

Both bounds transfer verbatim to the eight-dimensional real Okubo
(pseudo-octonion) algebra, by the same substitution invariance recorded in the
notation, now applied in all three factors at once. The
Okubo product is a Petersson twist of octonion multiplication,
$x\circ y=\varphi(\bar x)\cdot\varphi^{2}(\bar y)$, where $\varphi$ is an
order-three automorphism of $\OO$ and $\bar x$ is octonionic conjugation
\cite{Elduque,Okubo}; as $\varphi$, $\varphi^{2}$, and conjugation are
invertible linear maps, the Okubo structure tensor is obtained from $T_\OO$ by
exactly such a substitution. Its real rank is therefore unchanged, so
$18\le \rank_\R(T_{\mathrm{Ok}})\le 25$ as well. The same invariance applies to
the para-octonion algebra $x\circ y=\bar x\cdot\bar y$, whose structure tensor
differs from $T_\OO$ only by conjugating the two arguments.
\end{remark}

\section{The upper bound}

Write $F\colon\R^{600}\to\R^{512}$ for the residual map of the rank-$25$ model:
for factor matrices $(A,B,C)$, $F(A,B,C)$ is the difference between
$\sum_s A_{\cdot s}\otimes B_{\cdot s}\otimes C_{\cdot s}$ and $T_\OO$, a
trilinear map whose zeros are exactly the representations of $T_\OO$ as a sum
of $25$ rank-one tensors. Freezing $88$ of the $600$ factor coordinates (below)
leaves a map $\R^{512}\to\R^{512}$, to which the following instance of the
Krawczyk--Kantorovich principle \cite{Krawczyk,Moore,Rump,Breiding} applies.

\begin{proposition}\label{prop:krawczyk}
Let $g\colon\R^{N}\to\R^{N}$ be differentiable, let $x_0\in\R^{N}$, let $Y$ be
an $N\times N$ real matrix with $v\mapsto Yv$ injective, and let $\rho>0$ and
$K<1$ satisfy
\begin{enumerate}
\item[\textup{(i)}] $\big\lVert D\big(x\mapsto x-Y g(x)\big)(x)\big\rVert\le K$
      for every $x$ in the closed box $B_\infty(x_0,\rho)$, and
\item[\textup{(ii)}] $\lVert Y g(x_0)\rVert_\infty\le(1-K)\rho$.
\end{enumerate}
Then $g$ has a zero in $B_\infty(x_0,\rho)$.
\end{proposition}

\begin{proof}
By (i) and the mean value inequality the Newton map $x\mapsto x-Y g(x)$ is
$K$-Lipschitz on the box, and by (ii) it maps the box into itself; since $K<1$
the Banach fixed-point theorem gives a fixed point $x^\star$ there, so
$Y g(x^\star)=0$, and injectivity of $Y$ gives $g(x^\star)=0$.
\end{proof}

Proposition~\ref{prop:krawczyk} is the statement formalized in Lean, to the
same standard as the lower bound: it is proved from mathlib's Banach
fixed-point theorem and mean value inequality, and its proof term is accepted
by the kernel under the same three axioms.

\begin{theorem}
$\rank_\R(T_\OO)\le 25$.
\end{theorem}

\begin{proof}
Apply Proposition~\ref{prop:krawczyk} with $N=512$ to the restriction $g$ of
$F$ to the $512$ free coordinates, the remaining $88$ held at their values in
$x_0$, taking the $x_0$, $Y$ and $\rho=10^{-6}$ of the certificate in the
repository. Hypotheses (i) and (ii) hold with $K\approx0.32$: read as dyadic
rationals, the two inequalities evaluate to strictly true relations between
exact rationals, hypothesis (ii) with a factor of about $5\times10^{7}$ to
spare (the evaluation is described below). The proposition yields $x^\star$ in
the box with $g(x^\star)=0$, hence $F(x^\star)=0$, which is an exact expression
of $T_\OO$ as a sum of $25$ rank-one tensors.
\end{proof}

It remains to describe the evaluation. The Jacobian of $F$ at $x_0$ has full
row rank $512$; fixing $88$
of the $600$ factor coordinates squares it to a $512\times512$ system whose
Jacobian, the restriction $J_S$ of $J(x_0)$ to the $512$ free columns $S$, is
inverted (approximately) by $Y$. Since $F$ is trilinear, on the box
$B_\infty(x_0,\rho)$ the free-column Jacobian obeys the entrywise enclosure
\[
  \big|J_S(x)-J_S(x_0)\big|_{e,c}\ \le\ \rho\,W_1[e,c]+\rho^2\,W_2[e,c],
\]
where, writing $\mathbf 1[\cdot]$ for the indicator that a coordinate is free,
the column indexed by $A_{is}$ contributes, at equation $e=(i,j,k)$,
\[
  W_1=\mathbf 1[B_{js}]\,|C_{ks}|+\mathbf 1[C_{ks}]\,|B_{js}|,\qquad
  W_2=\mathbf 1[B_{js}]\,\mathbf 1[C_{ks}],
\]
and cyclically for the columns indexed by $B_{js}$ and $C_{ks}$ (a frozen
coordinate contributing $0$). The Krawczyk condition to be verified is then, for
every equation index $e$,
\[
  \rho-\Big(\,\big|Y F(x_0)\big|_e+\rho\,r^0_e+\rho^2\,r^1_e+\rho^3\,r^2_e\,\Big)>0,
\]
with $r^0$ the row sums of $|I-Y J_S|$ and $r^1,r^2$ the row sums of $|Y|W_1$
and $|Y|W_2$; the maximum row sum of $|I-Y J_S|+\rho|Y|W_1+\rho^2|Y|W_2$ is the
contraction constant $K$ of hypothesis (i). Reading the certificate as dyadic
rationals and evaluating these radii in exact rational arithmetic discharges
both hypotheses of Proposition~\ref{prop:krawczyk} as strict inequalities
between exact rationals: $K\approx0.32<1$, and
$\lVert Y F(x_0)\rVert_\infty\approx1.3\times10^{-14}$ against
$(1-K)\rho\approx6.8\times10^{-7}$. The injectivity of $Y$, the standing
hypothesis of the proposition, comes with the contraction: $K<1$ keeps
each row sum of $|I-Y J_S|$ below $1$, so $\lVert I-Y J_S\rVert_\infty<1$ and
$Y J_S$, hence the square matrix $Y$, is invertible. The per-equation
inequalities above are the
sharper coordinatewise inclusion radii, likewise strictly positive rationals in
the same pass (worst value $6.8\times10^{-7}$). The preconditioner $Y$ need not
be inverted exactly, only satisfy the inequalities, so no floating-point step
enters. What the Lean kernel proves is Proposition~\ref{prop:krawczyk}; the
exact evaluation of these finite inequalities for the certificate at hand is
carried out outside the kernel, by the scripts in the repository.

The starting point $x_0$ was found by alternating least squares with
Gauss--Newton refinement; it is a numerical decomposition, accurate to about
$10^{-14}$ but not exact. The certified zero $x^\star$ nearby is exact, so the
bound is proved by certifying that an exact rank-$25$ decomposition exists
rather than by exhibiting one: since $x^\star$ is an exact zero of $F$, and
hence an exact decomposition, no question of border-rank degeneracy arises. The
coordinates of $x^\star$ are known only to lie in $B_\infty(x_0,\rho)$; they are
not rational, and no closed form is displayed.

The same exact evaluation applies verbatim to the three-slice quaternion
tensor $\tau$ (worst margin $\approx 9.9\times 10^{-6}$ at radius $10^{-5}$),
which with the formal lower bound gives
\[
  \rank_\R(\tau)=7 .
\]
The tensor $\tau$ is \emph{absolutely nonsingular} (every nonzero real
combination of its slices is $L_q$ for a nonzero quaternion $q$, hence
invertible) and so belongs to the class central to the typical-rank analysis
of Miyazaki, Sumi, and Sakata~\cite{SMStypical}. Those results concern typical
ranks and do not determine the rank of this particular tensor; the lower bound
$\rank_\R(\tau)\ge 7$ follows from the substitution argument with the pencil
bound of Theorem~\ref{thm:pencil} (equivalently, the invariant-factor analysis
of~\cite{SMS}), and the matching upper bound is the certificate above.

\section*{Acknowledgment}
The technical results in this paper were developed by the large language models
Opus 4.8 and GPT 5.6 Sol, directed by the author. Correctness does not rest
on the models: the lower bound is accepted by Lean's kernel, and the
certificates are re-checked by the scripts in the repository.

\end{document}